\documentclass[11pt]{article}
\usepackage[margin=1in]{geometry}
\usepackage{graphicx}
\usepackage{amsmath}
\usepackage{amsthm}
\usepackage{amssymb}
\usepackage{parskip}
\usepackage{paralist}
\usepackage{url}
\usepackage{verbatim}

\usepackage[pagebackref,colorlinks=true,pdfpagemode=none,urlcolor=blue,linkcolor=blue,citecolor=blue,pdfstartview=FitH]{hyperref}

\usepackage{algorithmicx, algorithm}
\usepackage[noend]{algpseudocode}
\newcommand*\Let[2]{\State #1 = #2} 
\algrenewcommand\alglinenumber[1]{
    {\sf\footnotesize#1}}
\algrenewcommand\Return{\State \algorithmicreturn{} }%

\def\htto{\rightarrow}
\newif\ifcamera\cameratrue

\newtheorem{theorem}{Theorem}

\newtheorem{lemma}[theorem]{Lemma}
\newtheorem{claim}[theorem]{Claim}
\newtheorem{corollary}[theorem]{Corollary}
\newtheorem{definition}[theorem]{Definition}

\newtheorem{result}{Result}

\newcommand{\cA}{\mathcal{A}}

\newcommand{\cC}{\mathcal{C}}
\newcommand{\cD}{\mathcal{D}}

\newcommand{\cG}{\mathcal{G}}

\newcommand{\cK}{\mathcal{K}}

\newcommand{\cP}{\mathcal{P}}

\newcommand{\cT}{\mathcal{T}}

\newcommand{\cV}{\mathcal{V}}

\newcommand{\ga}{\alpha}
\newcommand{\gb}{\beta}
\newcommand{\gc}{\gamma}

\newcommand{\eps}{\epsilon}

\newcommand{\clean}{\mathop{clean}}
\newcommand{\trans}{\tau}

\newcommand{\Sec}[1]{\hyperref[sec:#1]{\S\ref*{sec:#1}}} 
\newcommand{\Eqn}[1]{\hyperref[eq:#1]{(\ref*{eq:#1})}} 
\newcommand{\Fig}[1]{\hyperref[fig:#1]{Fig.\,\ref*{fig:#1}}} 
\newcommand{\Tab}[1]{\hyperref[tab:#1]{Tab.\,\ref*{tab:#1}}} 
\newcommand{\Thm}[1]{\hyperref[thm:#1]{Theorem\,\ref*{thm:#1}}} 
\newcommand{\Lem}[1]{\hyperref[lem:#1]{Lemma\,\ref*{lem:#1}}} 
\newcommand{\Prop}[1]{\hyperref[prop:#1]{Prop.~\ref*{prop:#1}}} 
\newcommand{\Cor}[1]{\hyperref[cor:#1]{Corollary~\ref*{cor:#1}}} 
\newcommand{\Def}[1]{\hyperref[def:#1]{Definition~\ref*{def:#1}}} 
\newcommand{\Alg}[1]{\hyperref[alg:#1]{Alg.~\ref*{alg:#1}}} 
\newcommand{\Ex}[1]{\hyperref[ex:#1]{Ex.~\ref*{ex:#1}}} 
\newcommand{\Clm}[1]{\hyperref[clm:#1]{Claim~\ref*{clm:#1}}} 
\newcommand{\App}[1]{\hyperref[#1]{Appendix~\ref*{#1}}} 
\newcommand{\Res}[1]{\hyperref[res:#1]{Result\,\ref*{res:#1}}} 

\DeclareMathOperator{\poly}{poly}

\author{Rishi Gupta\thanks{Supported in part by the ONR PECASE Award
of the second author.}  \\\\
  Stanford University\\
  {\tt rishig@cs.stanford.edu}
\and Tim Roughgarden\thanks{
This research was supported in part by NSF Awards
CCF-1016885 and CCF-1215965, an AFOSR MURI grant, and an ONR PECASE Award.}
\\\\
Stanford University\\
{\tt tim@cs.stanford.edu}
\and C. Seshadhri \\\\
Sandia National Labs, CA\thanks{Sandia National Laboratories is a multi-program laboratory managed and operated by Sandia Corporation, a wholly owned subsidiary of Lockheed Martin Corporation, for the U.S. Department of Energy's National Nuclear Security Administration under contract DE-AC04-94AL85000.} \\
{\tt scomand@sandia.gov}}

\title{Decompositions of Triangle-Dense Graphs\thanks{A preliminary
version of this paper appeared in
the Proceedings of the 5th Innovations in
Theoretical Computer Science Conference, January 2014.}}
\date{}

\begin{document}

\maketitle

\begin{abstract}
High triangle density --- the graph property stating that a constant fraction of two-hop paths
belong to a triangle --- is a common signature
of social networks.  This paper studies triangle-dense
graphs from a structural perspective.  We prove constructively
that significant portions of a triangle-dense graph are contained in
a disjoint union of dense, radius $2$ subgraphs.  This result
quantifies the extent to which triangle-dense graphs resemble
unions of cliques.  We also show that our algorithm recovers
planted clusterings in approximation-stable $k$-median instances.
\end{abstract}

\section{Introduction} \label{sec:intro}

Can the special structure possessed by social networks
be exploited algorithmically?
Answering this question requires a formal definition of ``social
network structure.''
Extensive work on this topic has generated countless proposals but
little consensus (see e.g.~\cite{ChFa06}).
The most oft-mentioned (and arguably most validated) statistical
properties of social networks include heavy-tailed degree
distributions~\cite{BaAl99,BrKu+00,FFF99},
a high density of
triangles~\cite{WaSt98,SaCaWiZa10,UgKa+11}
and other dense subgraphs or
``communities''~\cite{For10,GiNe02,Ne03a,Ne06,LeLaDaMa08},
and low diameter
and the small world property~\cite{Kl00,Kl00-2,Kl02,Ne01}.

Much of the recent mathematical work on social networks has focused on
the important goal of developing generative models that produce random
networks with many of the above statistical properties.
Well-known examples of such models include preferential
attachment~\cite{BaAl99} and  related copying models~\cite{KuRa+00}, Kronecker
graphs~\cite{ChZhFa04,LeChKlFa10}, and the Chung-Lu random graph model~\cite{ChLu02,ChLu02-2}.
A generative model articulates a hypothesis about what ``real-world''
social networks look like, and is directly useful for generating
synthetic data.
Once a particular generative model of social networks is
adopted, a natural goal is to design algorithms tailored to perform
well on the instances generated by the model. It can also be used as a proxy
to study the effect of random processes (like edge deletions) on a network.
Examples of such results include~\cite{AlJeBa00,LiAm+08,MoSa10}.

This paper pursues a different approach.
In lieu of adopting a particular generative model for social networks,
we ask:
\begin{itemize}

\item []
{\em Is there a combinatorial assumption weak enough to hold in
  every ``reasonable'' model of social networks, yet strong enough to permit
  useful structural and algorithmic results?}

\end{itemize}
That is, we seek algorithms that offer non-trivial guarantees for {\em
  every} reasonable model of social networks, including those yet to
be devised.

\subsubsection*{Triangle-Dense Graphs}

We initiate the algorithmic study of {\em triangle-dense graphs}. Let a \emph{wedge} be a two-hop path in an undirected graph.
\medskip
\begin{definition}[Triangle-Dense Graph] \label{def:dense}
The {\em triangle density} of an undirected graph~$G=(V,E)$ is
$\trans(G) := 3t(G)/w(G)$, where $t(G)$ is the number of triangles
in~$G$ and~$w(G)$ is the number of wedges in~$G$ (conventionally, $\trans(G) = 0$ if $w(G)=0$).
The class of {\em $\eps$-triangle dense graphs} consists of the graphs~$G$
with $\trans(G) \geq \eps$.
\end{definition}
Since every triangle of a graph contains~3 wedges, and no two
triangles share a wedge, the triangle density of a graph is between~0
and~1.  We use the informal term ``triangle
dense graphs'' to mean graphs with constant triangle density.
In the social sciences, triangle density is usually
called the {\em transitivity} of a graph~\cite{WaFa94}.
We use the term triangle density because ``transitivity'' already has strong
connotations in graph theory.

As an example, the triangle density of a graph is 1 if and only if it is
the union of cliques.  The triangle density of an
Erd\"os-Renyi graph, drawn from~$G(n,p)$, is concentrated around~$p$.
Thus, only very dense Erd\"os-Renyi graphs have constant triangle
density.  Social networks are generally sparse and yet have remarkably
high triangle density; the Facebook graph, for instance, has triangle
density~$0.16$~\cite{UgKa+11}.
High triangle density is perhaps the least controversial signature of
social networks (see related work below).

The class of $\eps$-triangle dense graphs becomes quite diverse
as soon as $\eps$ is bounded below~1.
For example, the complete tripartite graph is triangle dense.
Every graph obtained from a bounded-degree graph by
replacing each vertex with a triangle is triangle dense.
Adding a clique on $n^{1/3}$ vertices
to a bounded-degree $n$-vertex graph produces a triangle-dense graph.
We give a litany of examples in \Sec{bad}.
Can there be interesting structural or algorithmic results for this
rich class of graphs?
\vspace{-10pt}
\subsubsection*{Our Results: A Decomposition Theorem}

Our main decomposition theorem quantifies the extent to which a
triangle-dense graph resembles a union of cliques.
The next definition gives our notion of an ``approximate union of
cliques.''  We use $G|_S$ to denote the subgraph of a graph~$G$
induced by a subset~$S$ of vertices.  Also, the \emph{edge density} of a
graph~$G=(V,E)$ is $|E|/\binom{|V|}{2}$.
\medskip
\begin{definition}[Tightly Knit Family] \label{def:part}
Let $\rho > 0$.
A collection $V_1, V_2, \ldots, V_k$ of disjoint sets of vertices of a
graph $G=(V,E)$ forms a \emph{$\rho$-tightly-knit family} if:
\begin{asparaitem}
	\item Each subgraph $G|_{V_i}$ has both edge density and triangle density at least $\rho$.
	\item Each subgraph $G|_{V_i}$ has radius at most $2$.
\end{asparaitem}
\end{definition}
When $\rho$ is a constant, we often refer simply to a tightly-knit family.
Every ``cluster'' of a tightly-knit family is dense in edges and in
triangles.
In the context of social networks, an abundance of triangles is generally
associated with meaningful social structure.

Our main decomposition theorem states that every triangle-dense graph
contains a tightly-knit family that captures a constant fraction of
the graph's triangles.
\medskip
\begin{result}[Main Decomposition Theorem] \label{res:main}
There exists a polynomial $f$ such that for every $\eps$-triangle dense graph~$G$,
there exists an $f(\eps)$-tightly-knit family that contains an $f(\eps)$ fraction
of the triangles of $G$.
\end{result}

We emphasize that \Res{main} requires only that the input
graph~$G$ is triangle dense --- beyond this property, it could be
sparse or dense, low- or high-diameter, and possess an arbitrary
degree distribution.  Graphs that are not triangle dense, such as
sparse Erd\"os-Renyi random graphs, do not generally admit non-trivial
tightly-knit families (even if the triangle density
requirement for each cluster is dropped).

Our proof of \Res{main} is constructive.  Using suitable data
structures, the resulting algorithm can be implemented to run in time
proportional to the number of wedges of the graph;
a detailed implementation is available from the authors.
%
This running time is reasonable for many social networks.
Our preliminary implementation
of the algorithm requires a few minutes on a commodity laptop to decompose networks
with millions of edges.
%

Note that \Res{main} is non-trivial only because we require that the tightly-knit family preserve the ``interesting social information'' of the original graph, in the form of the graph's triangles.
Extracting a single low-diameter cluster rich in edges and triangles
is easy --- large triangle density implies that typical vertex
neighborhoods have these properties.  But extracting such a cluster
carelessly can do more harm than good, destroying many triangles that
only partially intersect the cluster.
Our proof of \Res{main} shows how to repeatedly extract
low-diameter dense clusters while preserving at least a constant
fraction of the triangles of the original graph.

A triangle-dense graph need not contain a tightly-knit family that
contains a constant fraction of the graph's {\em edges}; see
the examples in~\Sec{bad}.  The culprit is that triangle density is a
``global'' condition and does not guarantee good local triangle
density everywhere, allowing room for a large number of edges that are
intuitively spurious.
Under the stronger condition of constant local triangle density, however,
we can compute a tightly-knit family with a stronger guarantee.
\medskip
\begin{definition}[Jaccard Similarity]
The {\em Jaccard similarity} of an edge~$e=(i,j)$ of a graph~$G=(V,E)$
is the fraction of vertices in the neighborhood of~$e$ that participate in triangles:
\begin{equation}\label{eq:jaccard}
J_e = \frac{|N(i) \cap N(j)|}{|N(i) \cup N(j) \setminus \{i,j\}|},
\end{equation}
where~$N(x)$ denotes the neighbors of a vertex~$x$ in~$G$.
\end{definition}
\medskip
\begin{definition}[Everywhere Triangle-Dense] \label{def:jaccard} A
  graph is \emph{everywhere $\eps$-triangle dense} if $J_e \geq \eps$
  for every edge~$e$, and there are no isolated vertices.
\end{definition}
Though useful conceptually, we would not expect graphs to be
everywhere triangle dense in practice.
The following weaker definition permits graphs that have a small
fraction of edges with low Jaccard similarity.
\medskip
\begin{definition}[$\mu,\eps$-Triangle-Dense] A
  graph is \emph{$\mu,\eps$-triangle dense} if $J_e \geq \eps$
  for at least a $\mu$ fraction of the edges~$e$.
\end{definition}
We informally refer to graphs with constant $\eps$ and high enough $\mu$ as \emph{mostly everywhere triangle dense}.
An everywhere $\eps$-triangle dense graph is $\mu,\eps$-triangle dense
for every $\mu$. An everywhere $\eps$-triangle dense graph is also
$\eps$-triangle dense.

The following is proved as \Thm{edges}.
\medskip
\begin{result}[Stronger Decomposition Theorem]\label{res:main-every}
There are polynomials $\mu, f$ such that for every $\mu(\eps),\eps$-triangle dense graph $G$, there exists an $f(\eps)$-tightly-knit family that contains an $f(\eps)$-fraction
of the \emph{edges and} triangles of $G$.
\end{result}

\paragraph{Applications to Planted Cluster Models.}

We give an algorithmic application of our decomposition in \Sec{bbg}, where
the tightly knit family produced by our algorithm is meaningful
in its own right.
We consider the
approximation-stable metric $k$-median instances introduced by Balcan,
Blum, and Gupta~\cite{BaBlGu13}.
By definition, every solution of an
approximation-stable instance that has
near-optimal objective function value is structurally similar to the
optimal solution. They reduce their problem to clustering a certain
graph with ``planted'' clusters corresponding to the optimal
solution. We prove that our algorithm recovers a close approximation
to the planted clusters, matching their guarantee.
%
%

\subsection{Discussion}

\paragraph{Structural Assumptions vs.\ Generative Models.}
Pursuing structural results and algorithmic guarantees that assume
only a combinatorial condition (namely, constant triangle density),
rather than a particular model of social networks, has clear
advantages and disadvantages.  The class of graphs generated by a
specific model will generally permit stronger structural
and algorithmic guarantees
than the class of graphs that share a single statistical property.
On the other hand, algorithms and results tailored to a
single model can lack robustness: they might not be meaningful if
reality differs from the model, and are less likely to translate
across different application domains that require different models.
Our results for triangle-dense graphs are relevant for every model of
social networks that generates such graphs with high probability, and
we expect that all future social network models will have this
property.  And of course, our results can be used in any application
domain that concerns triangle-dense graphs, whether motivated by
social networks or not.

Beyond generality and robustness, a second reason to prefer
a combinatorial assumption to a generative model is that the
assumption can be easily verified for a given data set.
Since computing the triangle density of a network is a well-studied
problem, both theoretically and practically
(see~\cite{SePiKo13} and the references therein), the extent to
which a network meets the triangle density assumption can be quantified.
By contrast, it is not clear how to argue that a network is a typical
instance from a generative model, other than by verifying various
statistical properties (such as triangle density).  This difficulty
of verification is amplified when there are multiple generative models
vying for prominence, as is currently the case with social and
information networks (e.g.~\cite{ChFa06}).

\noindent
\paragraph{Why Triangle Density?}
Social networks possess a number of statistical signatures, as
discussed above.  Why single out triangle density?
First, there is tremendous empirical support for large triangle
density in social networks.  This property has been studied for
decades in the social
sciences~\cite{HoLe70,Co88,Burt04,Fa06,FoDeCo10},
and recently there have been numerous large-scale studies on online
social networks~\cite{SaCaWiZa10,UgKa+11,SePiKo13}.
Second, in light of this empirical evidence, generative models for
social and information networks are explicitly designed to produce
networks with high
triangle density~\cite{WaSt98,ChFa06,SaCaWiZa10,ViBa12}.
Third, the assumption of constant triangle density seems to impose
more exploitable structure than the other most widely accepted
properties of social and information networks.
For example, the property of having small diameter
indicates little about the structure of a network --- every network
can be rendered small-diameter by adding one extra vertex connected to
all other vertices.  Similarly, merely assuming a power-law degree
distribution does not seem to impose significant restrictions on a
graph~\cite{FePaPa06}.  For example, the Chung-Lu model~\cite{ChLu02} generates
power-law graphs with no natural decompositions.
While constant triangle density is not a
strong enough assumption to exclude all ``obviously unrealistic
graphs,'' it nevertheless enables non-trivial decomposition results.
Finally, we freely admit that imposing one or more combinatorial
conditions other than triangle density could lead to equally
interesting results, and we welcome future work along such lines.
For example, recent work by Ugander, Backstrom, and
Kleinberg~\cite{UgBaKl13} suggests that constraining the frequencies
of additional small subgraphs could produce a refined model of social
and information networks.

\noindent
\paragraph{Why Tightly-Knit Families?}
We have intentionally highlighted the existence and computation of
tightly-knit families in triangle-dense graphs, rather than the
(approximate) solution of any particular computational problem
on such graphs.
Our main structural result quantifies the extent to which we can
``visualize'' a triangle-dense graph as, approximately, a union of
cliques.  This is a familiar strategy for understanding restricted
graph classes, analogous to using separator theorems to make precise
how planar graphs resemble grids~\cite{LiTa79}, tree decompositions to
quantify how
bounded-treewidth graphs resemble trees~\cite{RoSe86}, and the regularity
lemma to describe how dense graphs are approximately composed of
``random-like'' bipartite graphs~\cite{Sze78}.
Such structural results provide a flexible foundation for future
algorithmic applications.
We offer a specific application to recovering planted clusterings and leave
as future work the design of more applications.


\section{An intuitive overview} \label{sec:over}

We give an intuitive description of our proof. Our approach to finding a tightly-knit family
is an iterative extraction procedure. We find a single member of the family, remove this set from the graph (called the extraction),
and repeat. Let us start with an everywhere triangle-dense
graph $G$, and try to extract a single set $S$. It is easy to check that every vertex neighborhood is dense and has many triangles, and would qualify
as a set in a tightly-knit family.
But for vertex $i$, there may be many vertices outside $N(i)$ (the neighborhood of $i$)
that form triangles with a single edge contained in $N(i)$. By extracting $N(i)$, we could destroy too many triangles.
We give examples in \Sec{bad} where such a na\"{i}ve approach fails.

Here is a simple greedy fix to the procedure. We start by adding
$N(i)$ and $i$ to the set $S$. If any vertex
outside $N(i)$ forms many triangles with the edges in $N(i)$, we just add it to $S$. It is not clear that
we solve our problem by adding these vertices to $S$, since the extraction of $S$ could still destroy many triangles.
We prove that by adding at most $d_i$ vertices
(where $d_i$ is the
degree of $i$) with the highest number of triangles to $N(i)$, this
``destruction" can be bounded.
In other words, $G|_S$ will have a high density, obviously has radius $2$ (from $i$), and will contain a constant fraction
of the triangles incident to $S$.

Naturally, we can simply iterate this procedure and hope to get the entire tightly-knit family. But
there is a catch.
We crucially needed the graph to be \emph{everywhere} triangle-dense
for the previous argument. After extracting $S$,
this need not hold.  We therefore employ a \emph{cleaning} procedure that
iteratively removes edges of low Jaccard similarity
and produces an everywhere triangle-dense graph for the next extraction. This procedure also destroys some triangles,
but we can upper bound this number.
As an aside, removing low Jaccard similarity edges has been used for
sparsifying real-world graphs by Satuluri, Parthasarathy, and
Ruan~\cite{SaPaRu11}.

When the algorithm starts with an arbitrary triangle-dense graph $G$,
it first cleans the graph to get an everywhere triangle-dense graph.
We may lose many edges during the initial cleaning, and this is
inevitable, as examples in \Sec{bad} show.
In the end, this procedure constructs a tightly-knit family containing
a constant fraction of the triangles of the original triangle-dense graph.

When $G$ is everywhere or mostly everywhere triangle-dense, we can ensure that the
tightly-knit family contains a constant fraction
of the \emph{edges} as well.
Our proof is a non-trivial charging argument.  By assigning an appropriate
weight function to triangles and wedges, we can charge removed edges
to removed triangles. This (constructively) proves the existence
of a tightly-knit family with a constant fraction of edges and triangles.

\section{Extracting tightly-knit families} \label{sec:extract}

In this section we walk through the proof outlined in \Sec{over}
above. We first bound the losses from the cleaning procedure in
\Sec{clean}. We then show how to extract a member of a tightly-knit
family from a cleaned graph in \Sec{cluster}. We combine these two
procedures in \Thm{cftriangle} of \Sec{family} to obtain a full
tightly-knit family from a triangle-dense graph. Finally, \Thm{edges}
of \Sec{everywhere} shows that the procedure also preserves a constant
fraction of the edges in a mostly everywhere triangle-dense graph.

\subsection{Preliminaries} \label{sec:prelim}

We begin with some notation. Consider a graph $G = (V,E)$.
We index vertices with $i,j,k,\ldots$.
Vertex $i$ has degree $d_i$.
%
Let $S$ be a set of vertices. The number of triangles including some vertex
in $S$ is denoted $t_S$. We use $G|_S$ for the induced subgraph on $G$,
and $t^{(I)}_S$ for the number of triangles in $G|_S$ (the $I$ is for ``internal").
We repeatedly deal with subgraphs $H$ of $G$. We use the $\ldots(H)$
notation for the respective quantities in $H$. So, $t(H)$ would denote
the number of triangles in $H$, $d_i(H)$ denotes the degree of $i$ in $H$, etc.

\subsection{Cleaning a graph} \label{sec:clean}

An important ingredient in our constructive proof is a ``cleaning"
procedure that constructs an everywhere triangle-dense graph.
\medskip
\begin{definition} \label{def:clean} Consider the following procedure $\clean_\eps$ on a graph $H$
that takes input $\eps \in (0,1]$.
Iteratively remove an arbitrary edge with Jaccard similarity less than $\eps$,
as long as such an edge exists.
Finally, remove all isolated vertices.
We call this \emph{$\eps$-cleaning}, and denote the output
by $\clean_\eps(H)$.
\end{definition}
The output $\clean_\eps(H)$ is dependent on the order in which edges
are removed, but our results hold for an arbitrary removal order.
Satuluri et al.~\cite{SaPaRu11} use a more nuanced version of cleaning
for graph sparsification of social networks. They provide much
empirical evidence that removal of low Jaccard similarity
edges does not affect graph structure. Our arguments below may provide some
theoretical justification.
\medskip
\begin{claim} \label{clm:clean} The number of triangles in $\clean_\eps(H)$
is at least $t(H) - \eps w(H)$.
\end{claim}

\begin{proof}
The process $\clean_\eps$ removes a sequence of edges $e_1, e_2, \ldots$.
Let $W_l$ and $T_l$ be the set of wedges and triangles that are removed when $e_l$ is removed.
Since the Jaccard similarity of $e_l$ \emph{at this stage} is at most $\eps$, $|T_l| \leq \eps (|W_l|-|T_l|) \leq \eps |W_l|$.
All the $W_l$'s (and $T_l$'s) are disjoint.
Hence, the total number of triangles removed is $\sum_l |T_l| \leq \eps \sum_l |W_l| \leq \eps w(H)$.
\end{proof}
We get an obvious corollary by noting that $t(H) = \trans(H) \cdot w(H)/3$.
\medskip
\begin{corollary} \label{cor:clean} The graph $\clean_\eps(H)$ is everywhere $\eps$-triangle dense
and has at least $(\trans(H)/3 - \eps) w(H)$ triangles.
\end{corollary}

We also state a simple lemma on the properties of everywhere triangle-dense graphs.
\medskip
\begin{lemma}\label{lem:balance}
If $H$ is everywhere $\eps$-triangle dense, then $d_i \ge \eps\,d_j$ for every edge $(i,j)$.
Furthermore, $N(i)$ is $\eps$-edge dense for every vertex $i$.
\end{lemma}
\begin{proof}
If $d_i \ge d_j$ we are done.
Otherwise
\[ \eps \le J_{(i,j)} = \frac{|N(i) \cap N(j)|}{|(N(i)
  \setminus\{j\}) \cup (N(j)\setminus\{i\})|} \le \frac{d_i-1}{d_j-1}
\le \frac{d_i}{d_j}, \]
as desired.
To prove the second statement,
let $S = N(i)$.
The number of edges in $S$ is at least
\[ \frac{1}{2} \sum_{j\in S} |N(i) \cap N(j)| \ge \frac{1}{2} \sum_{j\in S} J_{(i,j)} (d_i -1) \ge \frac{\eps d_i(d_i - 1)}{2} = \eps \binom{d_i}{2}. \]
\end{proof}

\subsection{Finding a single cluster} \label{sec:cluster}

Suppose we have an everywhere $\eps$-triangle dense graph $H$. We show how to remove a single cluster
of a tightly-knit family. Since the entire focus of this subsection is on $H$, we drop
the $\ldots(H)$ notation.

For a set $S$ of vertices and $\rho \in (0,1]$, we say that $S$ is
  {\em $\rho$-extractable}
if: $H|_S$ is $\rho$-edge dense, $\rho$-triangle dense, $H|_S$ has radius $2$, and $t^{(I)}_S \geq \rho t_S$.
We define a procedure that finds a single
extractable cluster in the graph $H$.

\textbf{The extraction procedure:} Let $i$ be a vertex of maximum degree. For every vertex $j$, let $\theta_j$ be the number of triangles incident on $j$ whose other two vertices
are in $N(i)$. Let $R$ be the set of $d_i$ vertices with the largest $\theta_j$ values. Output $S = \{i\} \cup N(i) \cup R$.

It is not necessary to start with a vertex of maximum degree, but doing so provides a better dependence on $\eps$.
(\emph{Note:} Strictly speaking, the $\{i\}$ above is redundant; a simple
argument shows that $i\in R$.)

We start with a simple technical lemma.
\medskip
\begin{lemma}\label{lem:squares} Suppose $x_1 \ge x_2 \ge \cdots > 0$ with $\sum x_j \le \ga$ and $\sum x_j^2 \ge \gb$.
For all indices $r \le 2\ga^2/\gb$, $\sum_{j\le r} x_j^2 \ge \gb^2r/4\ga^2$.
\end{lemma}

\begin{proof}
If $x_{r+1} \ge \gb/2\ga$, then $\sum_{j\le r} x_j^2 \ge \gb^2r/4\ga^2$ as desired.
Otherwise,
\[ \sum_{j> r} x_j^2 \le x_{r+1}\sum_j x_j \leq \gb/2.\]
Hence, $\sum_{j \le r} x_j^2 = \sum x_j^2 - \sum_{j > r} x_j^2 \ge \gb/2 \ge \gb^2r/4\ga^2$, using the bound given for $r$.
\end{proof}
The main theorem of the section follows.
\medskip
\begin{theorem}\label{thm:extract} Let $H$ be an everywhere $\eps$-triangle dense graph.
The extraction procedure outputs an $\Omega(\eps^4)$-extractable set $S$ of vertices.
Furthermore, the number of edges in $H|_S$ is an $\Omega(\eps)$-fraction of the edges incident to $S$.
\end{theorem}

\begin{proof}
Let $\eps > 0$, $i$ a vertex of maximum degree, and $N = N(i)$.

We have $|S| \leq 2d_i$. By \Lem{balance}, $H|_N$ has at least $\eps
\binom{d_i}{2}$ edges,
so $H|_S$ is $\Omega(\eps)$-edge dense. By the size of $S$ and
maximality of $d_i$,
the number of edges in $H|_S$ is an $\Omega(\eps)$-fraction of the edges
incident to $S$.
It is also easy to see that $H|_S$ has radius 2.
It remains to show that $H|_S$ is $\Omega(\eps^4)$-triangle dense, and that $t_S^{(I)} = \Omega(\eps^4)t_S$.

For any $j$, let $\eta_j$ be the number of edges from $j$ to $N$, and let $\theta_j$ be the number of triangles incident on $j$ whose other two vertices are in $N$.
Let $x_j = \sqrt{2\theta_j}$.

\Lem{squares} tells us that if we can (appropriately) upper bound $\sum_j x_j$
and lower bound $\sum_j x_j^2$, then the sum of the largest few $x^2_j$'s
is significant. This implies that $H|_S$ has sufficiently many
triangles.
Using appropriate parameters, we show that
$H|_S$ contains $\Omega(\poly(\epsilon) \cdot d^3_i)$ triangles,
as opposed to trivial bounds that are quadratic in $d_i$.

%
%

\begin{claim} \label{clm:xj} We have $\sum_j x_j \leq \sum_{k \in N} d_k$, and
$\sum_{j} x_j^2 \geq \frac{\eps}{2} \sum_{k\in N} d_k(H|_N)\; d_k$, where $d_k(H|_N)$ is the degree of vertex $k$ within $H|_N$.
\end{claim}

\begin{proof}
We first upper bound $\sum_j x_j$:
\[ \sum_{j} x_j \le \sum_{j} \sqrt{2 \binom{\eta_j}{2}} \leq \sum_j
\eta_j = \sum_{k \in N} d_k. \]
The first inequality follows from $\theta_j \le
\binom{\eta_j}{2}$. The last equality is simply stating that the total number of edges to vertices in $N$ is the same as the total number of edges from vertices in $N$.


Let $t_e$ be the number of triangles that include the edge $e$.
For every $e = (k_1,k_2)$, $t_e \ge J_e\cdot \max(d_{k_1}-1, d_{k_2}-1)
\ge \eps\cdot \max(d_{k_1}-1, d_{k_2}-1)$. Since $\eps > 0$, each
vertex is incident on at least 1 triangle. Hence all degrees are at
least $2$, and $d_k-1 \ge d_k/2$ for all $k$. This means
\[ t_e \ge \frac{\eps\cdot\max(d_{k_1}, d_{k_2})}{2} \ge \frac{\eps(d_{k_1} + d_{k_2})}{4} \qquad \text{for all $e=(k_1,k_2)$}.\]

We can now lower bound $\sum_{j} x^2_j$. Abusing notation, $e \in H|_N$ refers to an edge in the induced subgraph. We have
\[ \sum_{j} x_j^2 = \sum_{j} 2\theta_j = \sum_{e \in H|_N} 2t_e \ge \sum_{(k_1,k_2) \in H|_N} \frac{\eps}{2} (d_{k_1} + d_{k_2}) = \frac{\eps}{2} \sum_{k\in N} d_k(H|_N)\; d_k. \]
The two sides of the second equality are counting (twice) the number of triangles ``to'' and ``from'' the edges of $N$.
%
\end{proof}

We now use \Lem{squares} with $\ga = \sum_{k\in N} d_k$, $\gb = \frac{\eps}{2} \sum_{k\in N} d_k(H|_N)\, d_k$, and $r = d_i$. We first check that $r \le 2\ga^2/\gb$. Note that $d_i \ge d_k \ge \eps d_i$ for all $k\in N$, by \Lem{balance} and by the maximality of $d_i$. Hence,
\[ \frac{2\ga^2}{\gb}
= \frac{4}{\eps}\,\frac{\left(\sum_{k\in N} d_k\right)^2}{\sum_{k\in N} d_k(H|_N)\,d_k}
\ge \frac{4}{\eps}\,\frac{\eps d_i|N|\sum_{k\in N} d_k}{d_i\sum_{k\in N} d_k}\ge 4d_i \ge r,\]
as desired. Let $R$ be the set of $r=d_i$ vertices with the highest
value of $\theta_j$, or equivalently, with the highest value of $x_j^2$. By
\Lem{squares}, $\sum_{j\in R} x_j^2 \ge \gb^2r/4\ga^2$, or $\sum_{j\in
  R} \theta_j \ge \gb^2r/8\ga^2$.
We compute
\[ \frac{\gb}{\ga}
= \frac{\eps}{2} \frac{\sum_{k\in N} d_k(H|_N)\; d_k}{\sum_{k\in N} d_k}
\ge \frac{\eps}{2} \min_{k\in N} d_k(H|_N)
\ge \frac{\eps^2d_i}{4}, \]
which gives $\sum_{j\in R} \theta_j \ge \eps^4d_i^3/128$. For the
first inequality above, think of the $d_k/\sum d_k$ as
the coefficients in a convex combination of
$d_k(H|_N)$'s. For the last inequality, $d_k(H|_N) = t_{(i,k)} \ge
J_{(i,k)} (d_i-1) \ge \eps d_i/2$ for all $k\in N$.

Recall $S = N \cup R$ and $|S| \le 2d_i$. We have
\[ t^{(I)}_S \ge \frac{\sum_{j \in R} \theta_j}{3} \ge \frac{\eps^4d^3_i}{384}, \]
since triangles contained in $N$ get overcounted by a factor of
3. Since both $t_S$ and the number of wedges in $S$ are bounded above
by $|S|\binom{d_i}{2} = \Theta(d_i^3)$, $H|_S$ is
$\Omega(\eps^4)$-triangle dense, and $t_S^{(I)} = \Omega(\eps^4)t_S$,
as desired.
\end{proof}

\subsection{Getting the entire family in a triangle-dense graph}
\label{sec:family}

We start with a
triangle-dense graph $G$ and explain how to get the desired entire
tightly-knit family.
Our procedure --- called the decomposition procedure --- takes as input a
parameter $\eps$.

\textbf{The decomposition procedure:} Clean the graph with
$\clean_\eps$, and run the extraction procedure to get a set
$S_1$. Remove $S_1$ from the graph, run $\clean_\eps$ again, and
extract another set $S_2$. Repeat until the graph is empty. Output the
sets $S_1, S_2, \ldots$.

We now prove our main theorem, \Res{main}, restated for convenience.

\medskip
\begin{theorem}\label{thm:cftriangle}
Consider a $\trans$-triangle dense graph $G$ and $\eps \leq \trans/4$.
The decomposition procedure outputs an $\Omega(\eps^4)$ tightly-knit
family with an $\Omega(\eps^4)$-fraction of the triangles of $G$.
\end{theorem}

\begin{proof}  We are guaranteed by \Thm{extract}
that $G|_{S_i}$ is $\Omega(\eps^4)$-edge and $\Omega(\eps^4)$-triangle
dense and has radius $2$.
It suffices to prove that an $\Omega(\eps^4)$-fraction of the triangles in $G$
are contained in this family.

Consider the triangles that are \emph{not} present in the tightly-knit family.
We call these the destroyed triangles.
Such triangles fall into two categories: those destroyed in the
cleaning phases, and those
destroyed when an extractable set is removed. Let $C$ be the triangles
destroyed during cleaning, and let $D_k$ be the  triangles destroyed in the $k$th extraction.
By the definition of extractable subsets and \Thm{extract}, $t(G|_{S_k}) = \Omega(\eps^4|D_k|)$.
Note that $C, D_k$, and the triangles in $G|_{S_k}$ (over all $k$) partition the
total set of triangles.
Hence, we get that $\sum_k t(G|_{S_k}) = \Omega(\eps^4 (t - |C|))$.

We now bound $|C|$. This follows the proof of \Clm{clean}. Let $e_1, e_2, \ldots$ be all the edges removed during cleaning phases.
Let $W_l$ and $T_l$ be the set of wedges and triangles that are destroyed when $e_l$ is removed.
Since the Jaccard similarity of $e_l$ at the time of removal is at most $\eps$, $|T_l| \leq \eps (|W_l|-|T_l|) \leq \eps |W_l|$.
All the $W_l$s (and $T_l$s) are disjoint.
Hence, $|C| = \sum_l |T_l| \leq \eps \sum_l |W_l| = \eps w = 3\eps
t/\trans \le 3t/4$, and $\sum_k t(G|_{S_k}) = \Omega(\eps^4 t)$, as
desired.
\end{proof}

\subsection{Preserving edges in a mostly everywhere triangle-dense graph}
\label{sec:everywhere}

For a mostly everywhere triangle-dense graph, we can also preserve a
constant fraction of the \emph{edges}.
This requires a more subtle argument. The aim of this subsection is to
prove the following (cf.\ \Res{main-every}).
\medskip
\begin{theorem} \label{thm:edges} Consider a $\mu,\gc$-triangle dense graph $G$, for $\mu \ge 1-\gc^2/32$.
The decomposition procedure, with $\eps \leq \gc^3/12$, outputs an
$\Omega(\eps^4)$ tightly-knit family with an $\Omega(\eps^4)$ fraction
of the triangles of $G$ and an $\Omega(\eps\gamma)$ fraction of the
edges of $G$.
\end{theorem}
The proof appears at the end of the subsection.
The tightly-knit family and triangle conditions follow directly from
\Thm{cftriangle}, so we focus on the edge condition.
By \Thm{extract}, the actual removal
of the clusters preserves a large enough fraction of the edges. The
difficulty is in bounding the edge removals during the cleaning phases.

We first give an informal description of the argument. We would like
to charge lost edges to lost triangles, and piggyback on the fact that
not many triangles are lost during cleaning. More specifically, we
apply a weight function to triangles (and wedges), such that losing or
keeping an edge corresponds to losing or keeping roughly one unit of
triangle (and wedge) weight in the graph. Most edges $(i,j)$ belong to
roughly $d_i+d_j$ triangles and wedges, and so intuitively we weight
each of those triangles (and wedges) by roughly $1/(d_i+d_j)$. This
intuition
breaks down if $d_i\ll d_j$, but $d_i\approx d_j$ for edges with high
Jaccard similarity.

The rest of the argument follows the high-level plan of the $\eps$-triangle
dense case (cf.\ the argument to bound $|C|$ in \Thm{cftriangle}),
though work is needed to replace triangles and wedges with
their weighted counterparts. The original graph $G$ has high triangle
density, which under our weight function is enough to imply a
comparable amount of triangle weight and wedge weight.
Only edges with low Jaccard similarity are removed during cleaning,
and each of these removed edges destroys significantly more wedge weight
than triangle weight. Hence, at the end of the process, a
lot of triangle weight must remain. There is a tight correspondence
between edges and triangle weight, and so a lot of edges must also
remain.

We now start the formal proof.
We use $E$, $W$, and $T$ to denote the sets of edges, wedges, and
triangles in $G$.
$W_e$ and $T_e$ denote the sets of edges and triangles that
include the edge $e$.
We use $E^c$, $W^c$, and $T^c$ to denote the respective sets destroyed
during the cleaning phases, and use $W^c_e$ and $T^c_e$ to denote the
corresponding local versions.
If an edge $e$ is removed during cleaning,
then $W^c_e \subseteq W_e$, but the
sets are not necessarily equal, since elements of $W_e$ may have been
removed prior to $e$ being cleaned.
Let $T^s = T\setminus T^c$.
Let $E^s$ and $V^s$ denote the edges and vertices, respectively,
included in at least one triangle of $T^s$.
For ease of reading, let $d'_i = d_i-1$ be one less than the degree of vertex $i$.

Call an edge $e$ \emph{good} if $J_e \ge \gc$ in the original
graph $G$, and \emph{bad} otherwise. We use $g_i$ to denote the number of good edges incident to vertex $i$.
Call a wedge good if it contains at least one good edge, and bad otherwise.
By hypothesis, a $\mu$ fraction of edges are good.
We make the following observation.

\medskip
\begin{claim} \label{clm:wbalance} For every good edge $(i,j)$, $d'_i
  \ge \gc d'_j$.
\end{claim}
\begin{proof}
We have
\[ \gc \le J_e = \frac{t_e}{d'_i+d'_j-t_e} \le \frac{d'_i}{d'_j}, \]
where the last inequality comes from $t_e \le d'_i$.
\end{proof}

We now define a \emph{weight} function $r$ on triangles and wedges.
For a triangle $\cT = (i_1,i_2,i_3)$ with at least 2 good edges,
define $r(\cT) = 1/d'_{i_1} + 1/d'_{i_2} + 1/d'_{i_3}$.
If $\cT$ has only one good edge $(i_1,i_2)$, then $r(\cT) = 1/d'_{i_1} +
1/d'_{i_2}$. If $\cT$ has no good edges, then $r(\cT) = 0$.
For a good wedge $w$ with central vertex $i$, $r(w) = 1/d'_i$, otherwise $r(w)=0$. Let $r(X) = \sum_{x\in X} r(x)$. Note that weights are always with respect to the degrees in the original graph $G$, and do not change over time.

In the next two claims we show that the total triangle weight in $G$
is comparable to the total wedge weight in $G$, and is also comparable
to $|E|$.

\medskip
\begin{claim} \label{clm:rtinit} $r(T) \ge \gc\mu |E|$. \end{claim}
\begin{proof}
Let $t^g_i$ be the number of triangles $(i,j,k) \in T$ for
which at least one of $(i,j),(i,k)$ is good.
Since the good edges each have Jaccard similarity $\ge \gc$, we have
$t^g_i \ge g_i\gc d'_i/2$.  Thus, \[ r(T) = \sum_{i} \frac{t^g_i}{d'_i}
\ge \sum_{i} \frac{g_i\gc}{2} = \gc\mu |E|. \qedhere \]
\end{proof}

\medskip
\begin{claim} \label{clm:rwinit} $r(W) \le 2\mu|E|$. \end{claim}

\begin{proof}
Let $w^g_i$ be the number of good wedges which have $i$ as their central vertex. Then
\[ r(W) = \sum_i \frac{w^g_i}{d'_i} \le \sum_i g_i = 2\mu|E|. \qedhere\]
\end{proof}

The next two claims bound the triangle weight lost by cleaning any particular edge.

\medskip
\begin{claim} \label{clm:greenclean} If a good edge $e = (i,j)$ is
  removed during cleaning, then $r(T^c_e) \leq (3\eps/\gc^2)r(W^c_e)$.
\end{claim}

\begin{proof}
Assume that
$d_i \ge d_j$. Let $d=d_i$.
We first lower bound $r(W^c_e)$ as a function of $|W^c_e|$. For any $w \in W^c_e$, $w$ has at least one good edge, and has either $i$ or $j$ as its central vertex. Hence $r(w) \ge \min\{1/d'_i,1/d'_j\}=1/d'$, and \[ r(W^c_e) \ge \frac{|W^c_e|}{d'}. \]

We now upper bound $r(T^c_e)$ as a function of $|T^c_e|$. Consider triangle $t = (i,j,k) \in T^c_e$. If $(i,j)$ is the only good edge in $t$, then $r(t) = 1/d'_i + 1/d'_j \le 2/d'\gc$, since $d'_j \ge d'\gc$ by \Clm{wbalance}. If $t$ has at least 2 good edges, then $k$ is at most 2 good edges away from $i$, and $d'_k \ge d'\gc^2$. This gives $r(t) = 1/d'_i+1/d'_j+1/d'_k \le 3/d'\gc^2$. Hence \[ r(T^c_e) \le \max\left\{\frac{3}{d'\gc^2},\frac{2}{d'\gc}\right\}|T^c_e| = \frac{3}{d'\gc^2}|T^c_e|. \]

Now, $|T^c_e| \le \eps |W^c_e|$, since $J_e\le \eps$ at the time of cleaning. Hence we have
\[ r(T^c_e) \le \frac{3}{d'\gc^2}|T^c_e| \le \frac{3\eps}{d'\gc^2}|W^c_e| \le \frac{3\eps}{\gc^2}\, r(W^c_e) \]
as desired.
\end{proof}

\medskip
\begin{claim} \label{clm:nongreenclean} If a bad edge $e = (i,j)$ is removed during cleaning, $r(T^c_e) \le 4/\gc$. \end{claim}
\begin{proof}
The only triangles with non-zero weight in $T^c_e$ have a good edge to $i$ and/or a good edge to $j$. Let $m_i$ and $m_j$ be the minimum degrees of any vertex connected by a good edge to $i$ and $j$, respectively. It is not too hard to see that
\[ r(T^c_e) \le g_i\left(\frac{1}{d'_i} + \frac{1}{m'_i}\right) + g_j\left(\frac{1}{d'_j} + \frac{1}{m'_j}\right). \]
Plugging in $m'_i \ge \gc d'_i$ (\Clm{wbalance}) and $g_i \le d'_i$ gives the desired result.
\end{proof}

We now combine the observations above to show that cleaning cannot remove all the triangle weight.

\medskip
\begin{claim} \label{clm:rtsgce} $r(T^s) \ge \gc|E|/4$. \end{claim}
\begin{proof}
We have
\begin{flalign*}
&& r(T^c) & = \sum_{\text{good } e} r(T^c_e) + \sum_{\text{bad } e} r(T^c_e) \\
&& & \le \sum_{\text{good } e} \frac{3\eps}{\gc^2}r(W_e^c) + \sum_{\text{bad } e} \frac{4}{\gc} && \text{by \Clm{greenclean} and \Clm{nongreenclean}} \\
&& & \le \frac{3\eps}{\gc^2}r(W) + \frac{4}{\gc}(1-\mu) |E| \\
&& & \le \frac{6\eps\mu|E|}{\gc^2} + \frac{4(1-\mu) |E|}{\gc} && \text{\phantom{\Clm{nongreenclean} and} by \Clm{rwinit}}\\
&& & \le \frac{\gc\mu|E|}{2} + \frac{\gc|E|}{8},
\end{flalign*}
where the last inequality follows from the bounds on $\eps$ and $\mu$ in the theorem statement. Hence
\begin{flalign*}
&& r(T^s) & = r(T)-r(T^c) \\
&& & \ge \gc\mu |E| - \left(\frac{\gc\mu|E|}{2} + \frac{\gc|E|}{8}\right) && \text{by \Clm{rtinit}} \\
&& & \ge \gc|E|/4,
\end{flalign*}
since $\mu \ge 3/4$.
\end{proof}

Finally, we show that if a subgraph of $G$ has high triangle weight, it must also have a lot of edges. Though the claim is stated in terms of $T^s$, the proof would hold for any $H\subset G$. This can be thought of as a moral converse to \Clm{rtinit}.

\medskip
\begin{claim} \label{clm:esrts} $r(T^s) \le |E^s|.$ \end{claim}
\begin{proof}
Let $H = (V,E^s)$. The triangles of $H$ are exactly $T^s$. We have
\[ r(T^s) \le \sum_{(i,j,k)\in T(H)} \frac{1}{d'_i(G)}+\frac{1}{d'_j(G)}+\frac{1}{d'_k(G)} = \sum_i \frac{t_i(H)}{d'_i(G)}, \]
where $t_i(H)$ is the number of triangles in $H$ incident to $i$. From here, we compute
\[\sum_i \frac{t_i(H)}{d'_i(G)} \le \frac{\binom{d_i(H)}{2}}{d'_i(G)} \le \sum \frac{d_i(H)}{2} = |E^s| \]
as desired.
\end{proof}

The last two claims together imply that the cleaning phase does not destroy too many edges. The rest of the proof is nearly identical to that of  \Thm{cftriangle} from the $\eps$-triangle dense case.

\begin{proof} (of \Thm{edges})
As noted above, the tightly-knit family and triangle conditions
follow directly from \Thm{cftriangle}.

Let $D_k$ be the edges destroyed in the $k$th
extraction, and let $E_k$ be the edges in $G|_{S_k}$. By
\Thm{extract}, $|E_k| = \Omega(\eps |D_k|)$. Since $E^c$, $D_k$, and
$E_k$ (over all $k$) partition $E$, we have $\sum_k |E_k| =
\Omega(\eps(|E|-|E^c|))$. Since $|E^c| + |E^s| \le |E|$, we have $\sum_k |E_k| = \Omega(\eps |E^s|)$. Finally, by \Clm{rtsgce} and \Clm{esrts}, $|E^s| = \Omega(\gc |E|)$, and so $\sum_k |E_k| = \Omega(\eps \gc |E|)$ as desired.
\end{proof}

\section{Triangle-dense graphs: the rogues' gallery} \label{sec:bad}

This section provides a number of examples of triangle-dense graphs.
These examples show, in particular, that radius-1 clusters are not
sufficient to capture a constant fraction of a triangle-dense graph's
triangles, and that tightly knit families cannot always capture a
constant fraction of a triangle-dense graph's edges.



\begin{asparaitem}

\item \textbf{Why radius 2?} Consider the complete tripartite graph. This
is everywhere triangle-dense. If we removed the $1$-hop neighborhood of any vertex,
we would destroy a $1-\Theta(1/n)$-fraction of the triangles. The only tightly-knit family
in this graph is the entire graph itself.

\item \textbf{More on 1-hop neighborhoods.} All 1-hop neighborhoods in an everywhere
triangle-dense graph are edge-dense, by \Lem{balance}. Maybe we could just take
the 1-hop neighborhoods of an independent set, to get a tightly-knit family? Of course,
the clusters would only be edge-disjoint (not vertex disjoint).

We construct an everywhere triangle-dense
graph where this does not work. There are $m+1$ disjoint sets of vertices,
$A_1, \dots, A_m, B$ each of size $m$. The graph induced on $\cup_k A_k$
is just a clique on $m^2$ vertices. Each vertex $b_k \in B$ is connected
to all of $A_k$. Note that $B$ is a maximal independent set, and the 1-hop neighborhoods
of $B$ contain $\Theta(m^4)$ triangles in total. However, the total number of triangles in the graph is $\Theta(m^6)$.

\item \textbf{Why we can't preserve edges.} \Res{main} only guarantees
  that the
tightly-knit family contains a constant fraction of the graph's
  triangles, not its edges.
Consider a graph that has a clique on $n^{1/3}$ vertices and an
arbitrary (or say, a random) constant-degree graph on the remaining $n-n^{1/3}$ vertices.
No tightly-knit family can involve vertices outside the clique, so most
of the edges must be removed. Of course, most edges in this case have
low Jaccard similarity.

\end{asparaitem}

In general, the condition of constant triangle density is fairly weak
and is met by a wide variety of graphs.
The following
two examples provide further intuition for this class of graphs.
\begin{asparaitem}
\item \textbf{A triangle-dense graph far from a disjoint union of cliques.}
Define the graph Bracelet$(m,d)$, for $m$ nodes of degree $d$, when $m > 4d/3$, as follows:
Let $B_1, \ldots, B_{3m/d}$ be sets of $d/3$ vertices each put in cyclic order. Note that $3m/d \ge 4$. Connect each vertex in $B_k$ to each vertex in $B_{k-1}, B_k$ and
$B_{k+1}$. Refer to \Fig{bracelent}. This is an everywhere triangle-dense $d$-regular graph on $m$ vertices.
Nonetheless, it is maximally far (i.e., $O(md)$ edges away) from a
disjoint union of cliques.
A tightly-knit family is obtained by taking $B_1 \cup B_2 \cup B_3$, $B_4 \cup B_5 \cup B_6$, etc.

\begin{figure}
\centering
\includegraphics{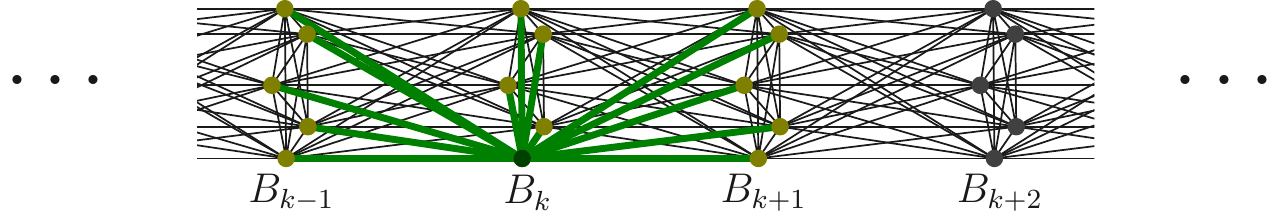}
\label{fig:bracelent}\caption{Bracelet graph with $d/3=5$. }
\end{figure}

\item \textbf{Hiding a tightly-knit family.} Start with $n/3$ disjoint
triangles. Now, add an arbitrary bounded-degree graph (say, an expander) on these
$n$ vertices. The resulting graph has constant triangle density, but most
of the structure is irrelevant for a tightly-knit family.
\end{asparaitem}

\section{Recovering a planted clustering} \label{sec:bbg}

This section gives an algorithmic application of our decomposition
procedure to recovering a ``ground truth'' clustering.
We study the planted clustering model defined by Balcan, Blum, and Gupta~\cite{BaBlGu13}, and show  that our algorithm gives guarantees
similar to theirs.
We do not subsume the results in \cite{BaBlGu13}.  Rather,
we observe that a graph problem that arises as a subroutine in their
algorithm is essentially that of finding a tightly-knit family in a
triangle-dense graph.
Their assumptions ensure that there is (up to minor perturbations) a
unique such family.

The main setting of~\cite{BaBlGu13} is as follows. Given a set of
points $V$ is some metric space, we wish to $k$-cluster them according
to some fixed objective function, such as the $k$-median objective.
Denote the optimal $k$-clustering by $\cC$
and the value by $OPT$.
The instance satisfies \emph{$(c,\epsilon)$-approximation-stability}
if for any $k$-clustering $\cC'$ of $V$ with objective
function value at most $c
\cdot OPT$, the ``classification distance" between $\cC$ and $\cC'$ is at most
$\epsilon$.
Thus, all solutions with near-optimal objective function value must be
structurally close to $\cC$.

In~\cite[Lemma 3.5]{BaBlGu13}, an approximation-stable $k$-median
instance is converted into a \emph{threshold graph}.
This graph $G = (V,E)$ contains $k$ disjoint cliques
$\{X_a\}_{a=1}^k$, such that the cliques do not have any common
neighbors.
These cliques correspond to clusters in the ground-truth clustering,
and their existence is a consequence of the approximation stability
assumption.
The aim is to get a $k$-clustering sufficiently close to
$\{X_a\}$. Formally, a $k$-clustering $\{S_a\}$ of $V$ is
\emph{$\Delta$-incorrect} if there is a permutation $\sigma$ such that
$\sum |X_{\sigma(a)} \setminus S_a| \le \Delta$.

Let $B = V \setminus \bigcup_a X_a$.
It is shown in~\cite{BaBlGu13} that when $|B|$ is small, good
approximations to $\{X_a\}$ can be found efficiently.
From our perspective, the interesting point is that when $|B|$ is much
smaller than $\sum_a |X_a|$,
the threshold graph has high triangle density.
Furthermore, as we prove below, the clusters
output by the extraction procedure of
\Thm{extract} are very close to the $X_a$'s of the threshold graph.

Suppose we want a $k$-clustering of a threshold graph. We iteratively
use the extraction procedure (from \Sec{cluster})
$k$ times to get clusters $S_1, S_2, \ldots, S_{k}$.
In particular, recall that at each step we choose a vertex $s_i$ with the current highest degree $d_i$. We set $N_i$ to be the $d_i$ neighbors of $s_i$ at this time, and $R$ to be the $d_i$ vertices with the largest number of triangles to $N_i$. Then, $S_i = \{i\} \cup N_i \cup R$.
The exact procedure of \Thm{cftriangle}, which includes cleaning, also works fine. Foregoing the cleaning step does necessitate a small technical change to the
extraction procedure: instead of adding all of $R$ to $S$, we only add
the elements of $R$ which have a positive number of triangles to
$N_i$.

We use the notation $N^*(U) = N(U) \cup U$. So $N^*(X_a) \cap N^*(X_b)
= \emptyset$, when $a \neq b$.
Unlike~\cite{BaBlGu13}, we assume that $|X_a|\ge 3$.
The following parallels the main theorem
of~\cite[Theorem 3.9]{BaBlGu13}, and the proof is similar to theirs.
\medskip
\begin{theorem} \label{thm:bbg} The output of the clustering algorithm above
is $O(|B|)$-incorrect on $G$.
\end{theorem}

\begin{proof}
We first map the algorithm's clustering to the true clustering
$\{X_a\}$.
Our algorithm outputs $k$ clusters, each with an associated ``center"
(the starting vertex).
These are denoted $S_1, S_2, \ldots,$ with centers $s_1, s_2, \ldots,$
in order of extraction.
We determine if there exists some true cluster $X_a$, such that $s_1
\in N^*(X_a)$.
If so, we map $S_1$ to $X_a$. (Recall the $N^*(X_a)$'s are disjoint,
so $X_a$ is unique if it exists.)
If no $X_a$ exists, we simply do not map $S_1$.
We then perform this for $S_2, S_3, \ldots$, except that we do not map
$S_k$ if we would be mapping it to an $X_a$ that has previously been
mapped to.
We finally end up with a subset $P \subseteq [k]$,
such that for each $a \in P$, $S_a$ is mapped to some $X_{a'}$.
By relabeling the true clustering, we can assume that for all $a \in
P$, $S_a$ is mapped to $X_a$.
The remaining clusters (for $X_{a \notin P}$) can be labeled with an
arbitrary permutation of $[k]\setminus P$.

Our aim is to bound $\sum_{a} |X_a \setminus S_a|$ by $O(|B|)$.

We perform some simple manipulations.
\begin{eqnarray*}
	\bigcup_{a}\, (X_a \setminus S_a) & = & \bigcup_{a \in P} (X_a \setminus S_a) \cup \bigcup_{a \notin P} (X_a \setminus S_a) \\
	& = & \bigcup_{a \in P} (X_a \cap \bigcup_{b < a} S_b) \cup \bigcup_{a \in P} (X_a \setminus \bigcup_{b \leq a} S_b) \cup \bigcup_{a \notin P} (X_a \setminus S_a) \\
	& \subseteq & \bigcup_{a}\, (X_a \cap \bigcup_{b < a} S_b) \cup \bigcup_{a \in P} (X_a \setminus \bigcup_{b \leq a} S_b) \cup \bigcup_{a \notin P} X_a.
\end{eqnarray*}
So we get the following sets of interest.
\begin{asparaitem}

	\item $L_1 = \bigcup_{a} (X_a \cap \bigcup_{b < a} S_b) =
	\bigcup_b (S_b \cap \bigcup_{a > b} X_a)$ is the set of
	vertices that are ``stolen'' by clusters before $S_a$.

	\item $L_2 = \bigcup_{a \in P} (X_a \setminus \bigcup_{b \le
	a} S_b)$ is the set of vertices that are left behind when
	$S_a$ is created.

	\item $L_3 = \bigcup_{a \notin P} X_a$ is the set of vertices
	that are never clustered.

\end{asparaitem}

Note that $\sum_{a} |X_a \setminus S_a| = |\bigcup_a (X_a\setminus
S_a)| \le |L_1| + |L_2| + |L_3|.$ The proof is completed by showing
that $|L_1| + |L_2| + |L_3| = O(|B|)$.
This will be done through a series of claims.

We first state a useful fact.
\medskip
\begin{claim} \label{clm:nb} Suppose for some $b \in
  \{1,2,\ldots,k\}$, $s_b \in
  N(X_b)$. Then $N_b$ is partitioned
into $N_b \cap X_b$ and $N_b \cap B$.
\end{claim}

\begin{proof} Any vertex in $N_b \setminus X_b$
must be in $B$. This is because $N_b$ is contained in a two-hop
neighborhood from $X_b$, which cannot intersect any other $X_a$.
\end{proof}
%
\medskip
\begin{claim} \label{clm:l1} For any $b$, $|S_b \cap \bigcup_{a > b} X_a| \leq 6|S_b \cap B|$.
\end{claim}

\begin{proof} We split into three cases. For convenience, let $U$ be the
set of vertices $S_b \cap \bigcup_{a > b} X_a$. Recall that $|S_b| \le 2d_b$.
\begin{asparaitem}
	\item For some $c$, $s_b \in X_c$: Note that $c\le b$ by the relabeling of clusters.
	Observe that $S_b$ is contained in a two-hop neighborhood of $s_b$, and hence cannot intersect
	any cluster $X_a$ for $a \neq c$. Hence, $U$ is empty.

	\item For some (unique) $c$, $s_b \in N(X_c)$: Again, $c\le b$.
	By \Clm{nb}, $d_b = |N_b| = |N_b \cap X_c| + |N_b \cap B|$.
	Suppose $|N_b \cap B| \geq d_b/3$. Then $|S_b \cap B| \geq |N_b \cap B| \geq d_b/3$.
	We can easily bound $|S_b| \leq 2d_b \leq 6|S_b \cap B|$.

	Suppose instead $|N_b \cap B| < d_b/3$, and hence $|N_b \cap X_c| > 2d_b/3$.
Note that $|N_b \cap X_c|$ is a clique. Each vertex in $N_b \cap X_c$
	makes ${|N_b \cap X_c|-1 \choose 2} \ge \binom{\lfloor 2d_b/3\rfloor}{2}$ triangles in $N_b$. On the other hand,
	the only vertices of $N_b$ that any vertex in $X_a$ for $a \neq c$ can connect to is in $N_b\cap B$. This forms fewer than ${\lfloor d_b/3\rfloor \choose 2}$ triangles in $N_b$. If ${\lfloor d_b/3\rfloor \choose 2} >0$, then $\binom{\lfloor 2d_b/3\rfloor}{2} > {\lfloor d_b/3\rfloor \choose 2}$.

	Consider the construction of $S_b$. We take the top $d_b$ vertices with the most triangles to $N_b$,
	and say we insert them in decreasing order of this number. Note that in the modified version of the algorithm, we only insert them while this number is positive.
	Before any vertex of $X_a$ ($a \neq b$) is added, all vertices of $N_b \cap X_c$ must be added.
	Hence, at most $d_b - |N_b \cap X_c| = |N_b \cap B| \leq |S_b \cap B|$ vertices
	of $\cup_{a \neq b} X_a$ can be added to $S_b$. Therefore, $|U| \leq |S_b \cap B|$.

	\item The vertex $s_b$ is at least distance $2$ from every $X_c$: Note that $N_b \subseteq S_b \cap B$.
	Hence, $|S_b| \leq 2d_b \leq 2|S_b \cap B|$. \qedhere
\end{asparaitem}
\end{proof}
\medskip
%
%
\begin{claim} \label{clm:l2} For any $a \in P$, $|X_a \setminus \bigcup_{b \le a} S_b| \leq |S_a \cap B|$.
\end{claim}

\begin{proof} Since $a \in P$, either $s_a \in X_a$ or $s_a \in N(X_a)$.  Consider the situation of the algorithm after the first $a-1$ sets $S_1, S_2, \ldots, S_{a-1}$
are removed. There is some subset of $X_a$ that remains; call it $X'_a = X_a \setminus \bigcup_{b < a} S_b$.

Suppose $s_a \in X_a$. Since $X'_a$ is still a clique, $X'_a \subseteq N_a$, and $(X_a \setminus \bigcup_{b \le a} S_b)$ is empty.

Suppose instead $s_a \in N(X_a)$.
Because $s_a$ has maximum degree and $X'_a$ is a clique, $d_a \geq |X'_a|-1$. Note that $|X'_a \setminus S_a|$ is what we wish to bound,
and $|X'_a \setminus S_a| \leq |X'_a \setminus N_a|$. By \Clm{nb}, $N_a$ partitions into $N_a \cap X_a = N_a \cap X'_a$
and $N_a \cap B$. We have $|X'_a \setminus N_a| = |X'_a| - |N_a \cap X_a| \leq d_a+1 - |N_a \cap X_a| = |N_a \cap B|+1
\leq |S_a \cap B|$.
\end{proof}
%
\medskip
\begin{claim} \label{clm:l3} $|L_3| \leq |B| + |L_1|$.
\end{claim}

\begin{proof} Consider some $X_a$ for $a \notin P$. Look at the situation when $S_1, \ldots, S_{a-1}$ are removed.
There is a subset $X'_a$ (forming a clique) left in the graph. All the vertices in $X_a \setminus X'_a$ are contained in $L_1$.
By maximality of degree, $d_a \geq |X'_a|-1$. Furthermore, since $a \notin P$, $N_a \subseteq B$ implying $d_a \leq |S_a \cap B|-1$.
Therefore, $|X'_a| \leq |S_a \cap B|$.
We can bound $\bigcup_{a \notin P} (X_a \setminus X'_a) \subseteq L_1$, and $\sum_{a \not\in P} |X'_a| \leq |B|$, completing the proof.
\end{proof}

To put it all together, we sum the bound of \Clm{l1}
and \Clm{l2} over $b \in [k]$ and $a\in P$ respectively to get $|L_1|
\leq 6|B|$ and $|L_2| \leq |B|$.
\Clm{l3} with the bound on $|L_1|$ yields $|L_3| \leq 7|B|$,
completing the proof of \Thm{bbg}. \
\end{proof}

\section{Conclusions}

This paper
proposes a ``model-free'' approach to the analysis of social
and information networks.  We restrict attention to
graphs that satisfy a combinatorial condition --- constant triangle
density --- in lieu of adopting a particular generative model.
The goal of this approach is to develop structural and algorithmic
results that apply simultaneously to all reasonable models of social
and information
networks.  Our main result shows that constant triangle density
already implies significant graph structure: every graph that meets
this condition is, in a precise sense, well approximated by a disjoint
union of clique-like graphs.

Our work suggests numerous avenues for future research.
\begin{enumerate}

\item Can the dependence of the inter-cluster edge and triangle
  density on the original graph's triangle density be improved?

\item The relative frequencies of four-vertex subgraphs also exhibit
  special patterns in social networks --- for example, there are
  usually very few induced four-cycles~\cite{UgBaKl13}.  Is there an
  assumption about four-vertex induced subgraphs, in conjunction with
  high triangle density, that yields a stronger decomposition theorem?

\item Are there interesting additional conditions under which the
  decomposition into a tightly-knit family is essentially unique?

\item Which computational problems are easier for triangle-dense graphs
  than for arbitrary graphs?  Just as planar separator theorems lead
  to faster algorithms and better heuristics for planar graphs than
  for general graphs, we expect our decomposition theorem to be a useful
  tool in the design of algorithms for triangle-dense graphs.

\end{enumerate}

\subsubsection*{Acknowledgements}

We are grateful for the helpful comments provided by Jon Kleinberg,
Johan Ugander, and the anonymous ITCS reviewers.

\bibliographystyle{alpha}
\bibliography{refs}

\newcommand{\etalchar}[1]{$^{#1}$}
\begin{thebibliography}{FWVDC10}

\bibitem[AJB00]{AlJeBa00}
R.~Albert, H.~Jeong, and A.-L. Barab{\'a}si.
\newblock Error and attack tolerance of complex networks.
\newblock {\em Nature}, 406:378--382, 2000.

\bibitem[BA99]{BaAl99}
A.-L. Barabasi and R.~Albert.
\newblock Emergence of scaling in random networks.
\newblock {\em Science}, 286:509--512, 1999.

\bibitem[BBG13]{BaBlGu13}
M.-F. Balcan, A.~Blum, and A.~Gupta.
\newblock Clustering under approximation stability.
\newblock {\em Journal of the ACM}, 60(2):1068--1077, 2013.

\bibitem[BKM{\etalchar{+}}00]{BrKu+00}
A.~Broder, R.~Kumar, F.~Maghoul, P.~Raghavan, S.~Rajagopalan, R.~Stata,
  A.~Tomkins, and J.~Wiener.
\newblock Graph structure in the web.
\newblock {\em Computer Networks}, 33:309--320, 2000.

\bibitem[Bur04]{Burt04}
R.~S. Burt.
\newblock Structural holes and good ideas.
\newblock {\em American Journal of Sociology}, 110(2):349--399, 2004.

\bibitem[CF06]{ChFa06}
D.~Chakrabarti and C.~Faloutsos.
\newblock Graph mining: Laws, generators, and algorithms.
\newblock {\em ACM Computing Surveys}, 38(1), 2006.

\bibitem[CL02a]{ChLu02}
F.~Chung and L.~Lu.
\newblock The average distances in random graphs with given expected degrees.
\newblock {\em Proceedings of the National Academy of Sciences},
  99(25):15879--15882, 2002.

\bibitem[CL02b]{ChLu02-2}
F.~Chung and L.~Lu.
\newblock Connected components in random graphs with given degree sequences.
\newblock {\em Annals of Combinatorics}, 6:125--145, 2002.

\bibitem[Col88]{Co88}
J.~S. Coleman.
\newblock Social capital in the creation of human capital.
\newblock {\em American Journal of Sociology}, 94:95--120, 1988.

\bibitem[CZF04]{ChZhFa04}
D.~Chakrabarti, Y.~Zhan, and C.~Faloutsos.
\newblock {R-MAT}: A recursive model for graph mining.
\newblock In {\em SIAM Conference on Data Mining}, pages 442--446, 2004.

\bibitem[Fau06]{Fa06}
Katherine Faust.
\newblock Comparing social networks: Size, density, and local structure.
\newblock {\em Metodoloski zvezki}, 3(2):185--216, 2006.

\bibitem[FFF99]{FFF99}
M.~Faloutsos, P.~Faloutsos, and C.~Faloutsos.
\newblock On power-law relationships of the internet topology.
\newblock In {\em Proceedings of SIGCOMM}, pages 251--262, 1999.

\bibitem[For10]{For10}
S.~Fortunato.
\newblock Community detection in graphs.
\newblock {\em Physics Reports}, 486:75--174, 2010.

\bibitem[FPP06]{FePaPa06}
A.~Ferrante, G.~Pandurangan, and K.~Park.
\newblock On the hardness of optimization in power law graphs.
\newblock In {\em Proceedings of Conference on Computing and Combinatorics},
  pages 417--427, 2006.

\bibitem[FWVDC10]{FoDeCo10}
B.~Foucault~Welles, A.~Van~Devender, and N.~Contractor.
\newblock Is a friend a friend?: {Investigating} the structure of friendship
  networks in virtual worlds.
\newblock In {\em Extended Abstracts on Human Factors in Computing Systems},
  pages 4027--4032, 2010.

\bibitem[GN02]{GiNe02}
M.~Girvan and M.~Newman.
\newblock Community structure in social and biological networks.
\newblock {\em Proceedings of the National Academy of Sciences},
  99(12):7821--7826, 2002.

\bibitem[HL70]{HoLe70}
P.~W. Holland and S.~Leinhardt.
\newblock A method for detecting structure in sociometric data.
\newblock {\em American Journal of Sociology}, 76:492--513, 1970.

\bibitem[Kle00a]{Kl00}
J.~M. Kleinberg.
\newblock Navigation in a small world.
\newblock {\em Nature}, 406:845, 2000.

\bibitem[Kle00b]{Kl00-2}
J.~M. Kleinberg.
\newblock The small-world phenomenon: An algorithmic perspective.
\newblock In {\em Proceedings of the Symposium on Theory of Computing}, pages
  163--170, 2000.

\bibitem[Kle02]{Kl02}
J.~M. Kleinberg.
\newblock Small-world phenomena and the dynamics of information.
\newblock In {\em Advances in Neural Information Processing Systems}, volume~1,
  pages 431--438, 2002.

\bibitem[KRR{\etalchar{+}}00]{KuRa+00}
R.~Kumar, P.~Raghavan, S.~Rajagopalan, D.~Sivakumar, A.~Tomkins, and E.~Upfal.
\newblock Stochastic models for the web graph.
\newblock In {\em Proceedings of Foundations of Computer Science}, pages
  57--65, 2000.

\bibitem[LAS{\etalchar{+}}08]{LiAm+08}
H.~Lin, C.~Amanatidis, M.~Sideri, R.~M. Karp, and C.~H. Papadimitriou.
\newblock Linked decompositions of networks and the power of choice in {P}olya
  urns.
\newblock In {\em Proceedings of the Symposium on Discrete Algorithms}, pages
  993--1002, 2008.

\bibitem[LCK{\etalchar{+}}10]{LeChKlFa10}
J.~Leskovec, D.~Chakrabarti, J.~M. Kleinberg, C.~Faloutsos, and Z.~Ghahramani.
\newblock Kronecker graphs: An approach to modeling networks.
\newblock {\em Journal of Machine Learning Research}, 11:985--1042, 2010.

\bibitem[LLDM08]{LeLaDaMa08}
J.~Leskovec, K.~Lang, A.~Dasgupta, and M.~Mahoney.
\newblock Community structure in large networks: Natural cluster sizes and the
  absence of large well-defined clusters.
\newblock {\em Internet Mathematics}, 6(1):29--123, 2008.

\bibitem[LT79]{LiTa79}
R.~J. Lipton and R.~E. Tarjan.
\newblock A separator theorem for planar graphs.
\newblock {\em SIAM Journal on Applied Mathematics}, 36(2):177--189, 1979.

\bibitem[MS10]{MoSa10}
A.~Montanari and A.~Saberi.
\newblock The spread of innovations in social networks.
\newblock {\em Proceedings of the National Academy of Sciences},
  107(47):20196--20201, 2010.

\bibitem[New01]{Ne01}
M.~E.~J. Newman.
\newblock The structure of scientific collaboration networks.
\newblock {\em Proceedings of the National Academy of Sciences},
  98(2):404--409, 2001.

\bibitem[New03]{Ne03a}
M.~E.~J. Newman.
\newblock Properties of highly clustered networks.
\newblock {\em Physical Review E}, 68(2):026121, 2003.

\bibitem[New06]{Ne06}
M.~E.~J. Newman.
\newblock Finding community structure in networks using the eigenvectors of
  matrices.
\newblock {\em Physical Review E}, 74(3):036104, 2006.

\bibitem[RS86]{RoSe86}
N.~Robertson and P.~D. Seymour.
\newblock Graph minors {I}{I}{I}: Planar tree-width.
\newblock {\em Journal of Combinatorial Theory, Series B}, 36(1):49--64, 1986.

\bibitem[SCW{\etalchar{+}}10]{SaCaWiZa10}
A.~Sala, L.~Cao, C.~Wilson, R.~Zablit, H.~Zheng, and B.~Y. Zhao.
\newblock Measurement-calibrated graph models for social network experiments.
\newblock In {\em Proceedings of the World Wide Web Conference}, pages
  861--870. ACM, 2010.

\bibitem[SPK13]{SePiKo13}
C.~Seshadhri, A.~Pinar, and T.~G. Kolda.
\newblock Fast triangle counting through wedge sampling.
\newblock In {\em Proceedings of the SIAM Conference on Data Mining}, 2013.

\bibitem[SPR11]{SaPaRu11}
V.~Satuluri, S.~Parthasarathy, and Y.~Ruan.
\newblock Local graph sparsification for scalable clustering.
\newblock In {\em Proceedings of ACM SIGMOD}, pages 721--732, 2011.

\bibitem[Sze78]{Sze78}
E.~Szemer\'{e}di.
\newblock Regular partitions of graphs.
\newblock {\em Probl\`{e}mes combinatoires et th{\'e}orie des graphes},
  260:399--401, 1978.

\bibitem[UBK13]{UgBaKl13}
J.~Ugander, L.~Backstrom, and J.~Kleinberg.
\newblock Subgraph frequencies: Mapping the empirical and extremal geography of
  large graph collections.
\newblock In {\em Proceedings of World Wide Web Conference}, pages 1307--1318,
  2013.

\bibitem[UKBM11]{UgKa+11}
J.~Ugander, B.~Karrer, L.~Backstrom, and C.~Marlow.
\newblock The anatomy of the facebook social graph.
\newblock Technical Report 1111.4503, Arxiv, 2011.

\bibitem[VB12]{ViBa12}
J.~Vivar and D.~Banks.
\newblock Models for networks: A cross-disciplinary science.
\newblock {\em Wiley Interdisciplinary Reviews: Computational Statistics},
  4(1):13--27, 2012.

\bibitem[WF94]{WaFa94}
S.~Wasserman and K.~Faust.
\newblock {\em Social Network Analysis: Methods and Applications}.
\newblock Cambridge University Press, 1994.

\bibitem[WS98]{WaSt98}
D.~Watts and S.~Strogatz.
\newblock Collective dynamics of `small-world' networks.
\newblock {\em Nature}, 393:440--442, 1998.

\end{thebibliography}

\end{document}